\documentclass{llncs}

\usepackage[utf8]{inputenc}
\usepackage{booktabs}

\usepackage{url}                  
\usepackage{listings}          
\usepackage{enumitem}      
\usepackage{hyperref}   

\usepackage[normalem]{ulem}
\usepackage{amsmath}
\usepackage{amsfonts}
\usepackage{mathrsfs}
\usepackage{stmaryrd}
\usepackage[usenames,dvipsnames]{xcolor}
\usepackage{syntax}
\usepackage{semantic}
\usepackage{listings}
\usepackage{alltt}
\usepackage{color}
\usepackage{xspace}
\usepackage{graphicx}
\usepackage{algorithm}
\usepackage{algorithmicx}
\usepackage{algpseudocode}
\usepackage{cite}
\usepackage{relsize}
\usepackage{wrapfig}
\usepackage{booktabs}
            
\usepackage{fancyhdr} 
\renewcommand{\subsubsection}[1]{\paragraph{#1}}

\renewcommand{\|}{\:|\:}
\newcommand{\N}{\mathbb{N}}
\newcommand{\B}{\mathbb{B}}
\newcommand{\false}{0}
\newcommand{\true}{1}
\newcommand{\pmap}{\rightharpoonup}
\newcommand\restr[2]{{
  \left.\kern-\nulldelimiterspace 
  #1 
  \vphantom{\big|} 
  \right|_{#2} 
  }}
\newcommand*{\defeq}{\stackrel{\Delta}{=}}
\newcommand{\xor}{\oplus}

\renewcommand{\implies}{\Rightarrow}

\newcommand{\revlang}{{\sc Revs}\xspace}
\newcommand{\revs}{{\sc Revs}\xspace}
\newcommand{\rever}{{\sc ReVerC}\xspace}
\newcommand{\fsharp}{F\#\xspace}
\newcommand{\fstar}{F$^\star$\xspace}

\newcommand{\rlet}[3]{\textsf{let } #1 = #2 \textsf{ in } #3}
\newcommand{\rfun}[2]{\lambda #1.#2}
\newcommand{\rapply}[2]{(#1\; #2)}
\newcommand{\rseq}[2]{#1; #2}
\newcommand{\rif}[3]{\textsf{if } #1 \textsf{ then } #2 \textsf{ else } #3}
\newcommand{\rfor}[4]{\textsf{for } #1 \textsf{ in } #2..#3 \textsf{ do } #4}
\newcommand{\rassign}[2]{#1\leftarrow#2}
\newcommand{\rtrue}{1}
\newcommand{\rfalse}{0}
\newcommand{\runit}{\textsf{unit}}
\newcommand{\rxor}[2]{#1\oplus#2}
\newcommand{\rand}[2]{#1\land#2}
\newcommand{\ror}[2]{#1\lor#2}
\newcommand{\rnot}[1]{\neg#1}
\newcommand{\rregister}[2]{\textsf{reg } #1 \dots #2}
\newcommand{\rindex}[2]{#1.[#2]}
\newcommand{\rslice}[3]{#1.[#2..#3]}
\newcommand{\rappend}[2]{\textsf{append } #1\; #2}
\newcommand{\rrotate}[2]{\textsf{rotate } #1\; #2}
\newcommand{\rclean}[1]{\textsf{clean } #1}
\newcommand{\rassert}[1]{\textsf{assert } #1}

\newcommand{\nil}{\textsf{nil}}

\newcommand{\Heap}{\textbf{Store}}
\newcommand{\heap}{\sigma}
\newcommand{\dom}{\textsf{dom}}
\newcommand{\cod}{\textsf{cod}}
\newcommand{\subst}[2]{[#1\mapsto#2]}
\newcommand{\config}[2]{\langle #1, #2\rangle}

\newcommand{\sem}[1]{\llbracket #1\rrbracket}

\newcommand{\Anc}{\mathbf{AncHeap}}
\newcommand{\Circ}{\mathbf{Circ}}
\newcommand{\Stat}{\textnormal{\textbf{State}}}
\newcommand{\notgate}[1]{\text{NOT } #1}
\newcommand{\cnot}[2]{\text{CNOT } #1\; #2}
\newcommand{\toffoli}[3]{\text{Toffoli } #1\; #2\; #3}
\newcommand{\bcomp}{\textsc{compile-BExp}}

\newcommand{\BExp}{\textnormal{\textbf{BExp}}}
\newcommand{\vars}[1]{\mathsf{vars}(#1)}
\newcommand{\ah}{\xi}
\renewcommand{\nil}{-}

\newcommand{\I}{\mathscr{I}}

\newcommand{\assign}{\textnormal{\textsf{assign}}}
\renewcommand{\eval}{\textnormal{\textsf{eval}}}
\newcommand{\cl}{\kappa}

\newcommand{\uses}[1]{\mathsf{use}(#1)}
\newcommand{\mods}[1]{\mathsf{mod}(#1)}
\newcommand{\controls}[1]{\mathsf{control}(#1)}
\newcommand{\st}{s}

\newcommand{\uncompute}[2]{\mathsf{uncompute}(#1, #2)}

\begin{document}

\title{Verified compilation of space-efficient reversible circuits}

\author{
  Matthew Amy\inst{1, 2} \and Martin Roetteler\inst{3} \and Krysta M.~Svore\inst{3}
}
\institute{
  Institute for Quantum Computing, Waterloo, Canada
  \and David R. Cheriton School of Computer Science, University of Waterloo, \\ Waterloo, Canada
  \and Microsoft Research, Redmond, USA
}
	
\maketitle

\begin{abstract}
The generation of reversible circuits from high-level code is an important problem in several application domains, including low-power electronics and quantum computing. Existing tools compile and optimize reversible circuits for various metrics, such as the overall circuit size or the total amount of space required to implement a given function reversibly. However, little effort has been spent on verifying the correctness of the results, an issue of particular importance in quantum computing. There, compilation allows not only mapping to hardware, but also the estimation of resources required to implement a given quantum algorithm, a process that is crucial for identifying which algorithms will outperform their classical counterparts. We present a reversible circuit compiler called \rever, which has been formally verified in F$^\star$ and compiles circuits that operate correctly with respect to the input program. Our compiler compiles the \revs language \cite{Parent:2015} to combinational reversible circuits with as few ancillary bits as possible, and provably cleans temporary values.
\end{abstract}

\section{Introduction}

The ability to evaluate classical functions coherently and in superposition as part of a larger quantum computation is essential for many quantum algorithms. For example, Shor's quantum algorithm \cite{Shor:97} uses classical modular arithmetic and Grover's quantum algorithm \cite{Grover:96} uses classical predicates to implicitly define the underlying search problem. There is a resulting need for tools to help a programmer translate classical, irreversible programs into a form which a quantum computer can understand and carry out, namely into reversible circuits, which are a special case of quantum transformations \cite{NC:2000}. Other applications of reversible computing include low-power design of classical circuits. See \cite{Markov:2014} for background and a critical discussion.

Several tools have been developed for synthesizing reversible circuits, ranging from low-level methods for small circuits such as \cite{MMD:2003,MMD:2007,WD:2010,SSP:2013,LJ:2014} (see also \cite{SM:2013} for a survey) to high-level programming languages and compilers \cite{yokoyama2007reversible,Wille:2010,Thomsen:2012,GLR+:2013b,Parent:2015}. In this paper we are interested in the latter class---i.e., methods for compiling high-level code to reversible circuits. Such compilers commonly perform optimization, as the number of bits quickly grows with the standard techniques for achieving reversibility (see, e.g., \cite{Scherer:2015}). The question, as with general purpose compilers, is whether or not we can trust these optimizations.

In most cases, extensive testing of compiled programs is sufficient to establish the correctness of both the source program and its translation to a target architecture by the compiler. Formal methods are typically reserved for safety- (or mission-) critical applications.  For instance, formal verification is an essential step in modern computer-aided circuit design due largely to the high cost of a recall. Reversible -- specifically, quantum -- circuits occupy a different design space in that 1) they are typically ``software circuits,'' i.e., they are not intended to be implemented directly in hardware, and 2) there exist few examples of hardware to actually run such circuits. Given that there are no large-scale universal quantum computers currently in existence, one of the goals of writing a quantum circuit compiler at all is to accurately gauge the amount of physical resources needed to perform a given algorithm, a process called \emph{resource estimation}. Such resource estimates can be used to identify the ``crossover point'' when a problem becomes more efficient to solve on a quantum computer, and are invaluable both in guiding the development of quantum computers and in assessing their potential impact. However, different compilers give wildly different resource estimates for the same algorithms, making it difficult to trust that the reported numbers are correct. For this reason compiled circuits need to have some level of formal guarantees as to their correctness for resource estimation to be effective. 

In this paper we present \rever, a lightly optimizing compiler for the \revs language \cite{Parent:2015} which has been written and proven correct in the dependently typed language \fstar \cite{Swamy:2015}. Circuits compiled with \rever are certified to preserve the semantics of the source \revs program, which we have for the first time formalized, and to reset or \emph{clean} all ancillary (temporary) bits used so that they may be used later in other computations. In addition to formal verification of the compiler, \rever provides an assertion checker which can be used to formally verify the source program itself, allowing effective end-to-end verification of reversible circuits. 

\subsubsection{Contributions} The following is a summary of the contributions of our paper:
\begin{itemize}
\item We give a formal semantics of \revlang.
\item We present a compiler for \revlang called \rever, written in \fstar. The compiler currently has three modes: direct to circuit, eager-cleaning, and Boolean expression compilation. 
\item We develop a new method of eagerly cleaning bits to be reused again later, based on \emph{cleanup expressions}.
\item Finally, we verify correctness of \rever with machine-checked proofs that the compiled reversible circuits faithfully implement the input program's semantics, and that all ancillas used are returned to their initial state.
\end{itemize}

\subsubsection{Related work}

Due to the reversibility requirement of quantum computing, quantum programming languages and compilers typically have methods for generating reversible circuits. Quantum programming languages typically allow compilation of classical, irreversible code in order to minimize the effort of porting existing code into the quantum domain. In QCL \cite{Omer:2000}, ``pseudo-classical'' operators -- classical functions meant to be run on a quantum computer -- are written in an imperative style and compiled with automatic ancilla management. As in \revs, such code manipulates registers of bits, splitting off sub-registers and concatenating them together. The more recent Quipper \cite{GLR+:2013b} automatically generates reversible circuits from classical code by a process called \emph{lifting}: using Haskell metaprogramming, Quipper lifts the classical code to the reversible domain with automated ancilla management. However, little space optimization is performed \cite{Scherer:2015}. 

Verification of reversible circuits has been previously considered from the viewpoint of checking equivalence against a benchmark circuit or specification \cite{Wille:2009,Yamashita:2010}. This can double as both \emph{program verification} and \emph{translation validation}, but every compiled circuit needs to be verified separately. Moreover, a program that is easy to formally verify may be translated into a circuit with hundreds of bits, and is thus very difficult to verify. Recent work has shown progress towards verification of more general properties of reversible and quantum circuits via model checking \cite{Zuliani16}, but to the authors' knowledge, no verification of a reversible circuit compiler has yet been carried out. By contrast, many compilers for general purpose programming languages have been formally verified in recent years -- most famously, the CompCert optimizing C compiler \cite{Leroy:2006}, written and verified in Coq. Since then, many other compilers have been developed and verified in a range of languages and logics including Coq, HOL, \fstar, etc., with features such as shared memory \cite{Beringer:2014}, functional programming \cite{Chlipala:2010,Fournet:2013} and modularity \cite{Perconti:2014,Neis:2015}.

\section{Reversible computing}\label{sec:overview}

Reversible functions are Boolean functions $f:\{0,1\}^n \rightarrow \{0,1\}^n$ which can be inverted on all outputs, i.e., precisely those functions which correspond to permutations of a set of cardinality $2^n$, for some $n \in \N$. As with classical circuits, reversible functions can be constructed from universal gate sets -- for instance, it is known that the Toffoli gate which maps $(x,y,z) \mapsto (x,y,z\oplus (x\land y))$, together with the controlled-NOT gate (CNOT) which maps $(x,y) \mapsto (x, x \oplus y)$ and the NOT gate which maps $x \mapsto x \oplus 1$, is universal for reversible computation \cite{NC:2000}.

An important metric that is associated with a reversible circuit is the amount of scratch space required to implement a given target function, i.e., temporary bits which store intermediate results of a computation. In quantum computing such bits are commonly denoted as {\em ancilla} bits. A very important difference to classical computing is that scratch bits cannot just be overwritten when they are no longer needed: any ancilla that is used as scratch space during a reversible computation must be returned to its initial value---commonly assumed to be $0$---computationally. Moreover, if an ancilla bit is not ``cleaned'' in this way, in a quantum computation it may remain entangled with the computational registers which in turn can destroy the desired interferences that are crucial for many quantum algorithms. 

Figure \ref{fig:exampleCircuit} shows a reversible circuit over NOT, CNOT, and Toffoli gates computing the NOR function. Time flows left to right, with input values listed on the left and outputs listed on the right. NOT gates are depicted by an $\oplus$, while CNOT and Toffoli gates are written with an $\oplus$ on the \emph{target} bit (the bit whose value changes) connected by a vertical line to, respectively, either one or two \emph{control} bits identified by solid dots. Two ancilla qubits are used which both initially are $0$; one of these ultimately holds the (otherwise irreversible) function value, the other is returned to zero. For larger circuits, it becomes a non-trivial problem to assert a) that indeed the correct target function $f$ is implemented and b) that indeed all ancillas that are not outputs are returned to $0$.

\subsubsection{\revlang}

\begin{wrapfigure}{r}{0.5\textwidth}
\centerline{\includegraphics[scale=0.3]{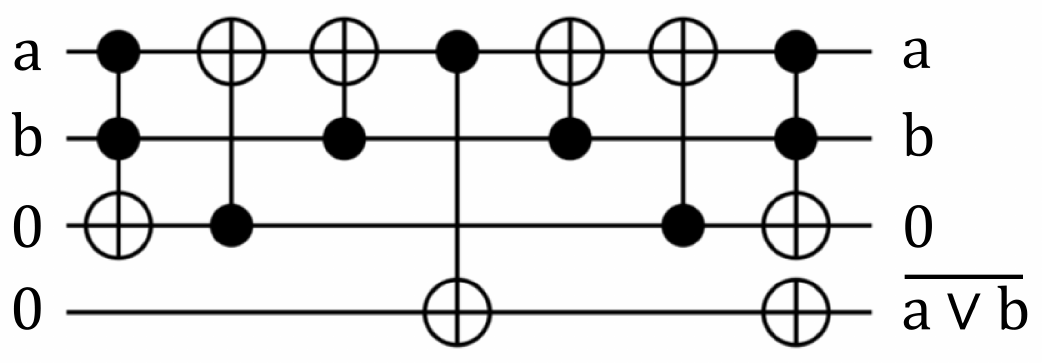}}\bigskip
\caption{A Toffoli network computing the NOR function $f(a,b) = \overline{a\vee b}$.}
\label{fig:exampleCircuit}
\end{wrapfigure}

From Bennett's work on reversible Turing machines it follows that any function can be implemented by a suitable reversible circuit \cite{Bennett:73}: if an $n$-bit function $x \mapsto f(x)$ can be implemented with $K$ gates over $\{{\rm NOT}, {\rm AND}\}$, then the reversible function $(x,y) \mapsto (x,y\oplus f(x))$ can be implemented with at most $2K+n$ gates over the Toffoli gate set. The basic idea behind Bennett's method is to replace all AND gates with Toffoli gates, then perform the computation, copy out the result, and undo the computation. This strategy is illustrated in Figure \ref{fig:bennett}, where the box labelled $U_f$ corresponds to $f$ with all AND gates substituted with Toffoli gates and the inverse box is simply obtained by reversing the order of all gates in $U_f$. Bennett's method has been used to perform classical-to-reversible circuit compilation in the quantum programming language Quipper \cite{GLR+:2013b}. One potential disadvantage of Bennett's method is the large number of ancillas it requires as the required memory scales proportional to the circuit {\em size} of the initial, irreversible function $f$. 

\begin{wrapfigure}[12]{l}{0.5\textwidth}
\centerline{\includegraphics[scale=0.3]{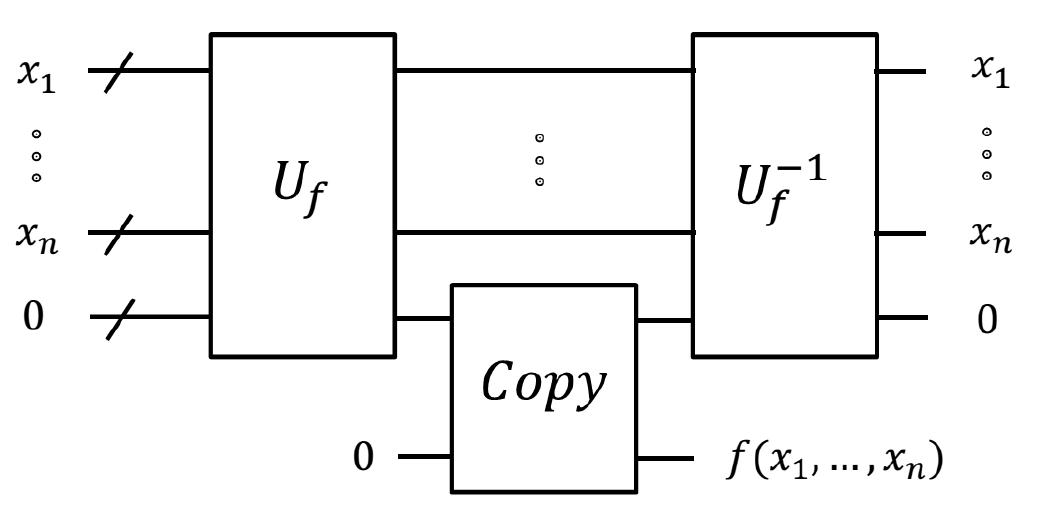}}
\caption{A reversible circuit computing $f(x_1,\dots, x_m)$ using the Bennett trick. Input lines with slashes denote an arbitrary number of bits.}
\label{fig:bennett}
\end{wrapfigure}

In recent work, an attempt was made with the \revs compiler (and programming language of the same name) \cite{Parent:2015} to improve on the space-complexity of Bennett's strategy by generating circuits that are \emph{space-efficient} -- that is, \revs is an optimizing compiler with respect to the number of bits used. Their method makes use of a dependency graph to determine which bits may be eligible to be cleaned \emph{eagerly}, before the end of the computation and hence be reused again. We build on their work in this paper, formalizing \revs and developing a \emph{verified} compiler without too much loss in efficiency. In particular, we take the idea of eager cleanup and develop a new method analogous to garbage collection.

\section{Languages}\label{sec:language}
In this section we give a formal definition of \revlang, as well as the intermediate and target languages of the compiler.

\subsection{The Source}
The abstract syntax of \revlang is presented in Figure~\ref{fig:syntax}. The core of the language is a simple imperative language over Boolean and array (register) types. The language is further extended with ML-style functional features, namely first-class functions and \emph{let} definitions, and a reversible domain-specific construct \emph{clean} which asserts that its argument evaluates to $0$ and frees a bit.

\begin{figure}[t]
	\footnotesize
	\begin{align*}
	{\bf Var} \;&x, \quad {\bf Bool} \;b\in\{0,1\}=\B, \quad{\bf Nat} \;i, j\in\N, \quad {\bf Loc} \;l\in\N \\
	{\bf Val} \;&v ::= \runit \| l \| \rregister{l_1}{l_n} \| \rfun{x}{t} &\\
		{\bf Term} \; &t ::= \rlet{x}{t_1}{t_2} \| \rfun{x}{t} \| \rapply{t_1}{t_2} \| \rseq{t_1}{t_2} \| x \| \rassign{t_1}{t_2} \| 
		         b \| \rxor{t_1}{t_2} \|  \rand{t_1}{t_2} & \\
			&\hspace*{1.7em}|\; \rclean{t} \| \rassert{t} \| \rregister{t_1}{t_n}\! \| \rindex{t}{i} \| \rslice{t}{i}{j} \| 
			\rappend{t_1\!}{t_2}\! \| \rrotate{i}{t} &
	\end{align*}
\caption{Syntax of \revlang.}
\label{fig:syntax}
\end{figure}

In addition to the basic syntax of Figure~\ref{fig:syntax} we add the following derived operations:
\vspace{-10pt}
\begin{gather*}
	\rnot{t} \defeq \rtrue\oplus t, \qquad \ror{t_1}{t_2} \defeq \rxor{(\rand{t_1}{t_2})}{(\rxor{t_1}{t_2})}, \\
	\rif{t_1}{t_2}{t_3}\defeq \rxor{(\rand{t_1}{t_2})}{(\rand{\rnot{t_1}}{t_3})}, \\
	\rfor{x}{i}{j}{t}\defeq t\subst{x}{i};\cdots;t\subst{x}{j}.
\end{gather*}
Note that \revlang has no \emph{dynamic} control -- i.e. control dependent on run-time values. In particular, every \revlang program can be transformed into a straight-line program. This is due to the restrictions of our target architecture (see below). 

The \rever compiler uses \fsharp as a meta-language to generate \revs code with particular register sizes and indices, possibly computed by some more complex program. Writing an \fsharp program that generates \revlang code is similar in effect to writing in a hardware description language \cite{Claessen:2001}. We use \fsharp's \emph{quotations} mechanism to achieve this by writing \revlang programs in quotations \texttt{<@\dots @>}. Note that unlike languages such as Quipper, our strictly combinational target architecture doesn't allow computations in the meta-language to depend on computations within \revs.

\begin{example}
\begin{figure}[t]
\begin{lstlisting}[mathescape=true]
fun a b ->
  let carry_ex a b c = (a ${\color{BurntOrange}\land}$ (b ${\color{BurntOrange}\oplus}$ c)) ${\color{BurntOrange}\oplus}$ (b ${\color{BurntOrange}\land}$ c)
  let result = Array.zeroCreate(n)
  let mutable carry = false
  
  result.[0] ${\color{BurntOrange}\leftarrow}$ a.[0] ${\color{BurntOrange}\oplus}$ b.[0]
  for i in 1 .. n-1 do
    carry ${\color{BurntOrange}\leftarrow}$ carry_ex a.[i-1] b.[i-1] carry
    result.[i] ${\color{BurntOrange}\leftarrow}$  a.[i] ${\color{BurntOrange}\oplus}$ b.[i] ${\color{BurntOrange}\oplus}$ carry
  result
\end{lstlisting}
\caption{Implementation of an $n$-bit adder.}
\label{fig:add}
\end{figure}

Figure~\ref{fig:add} gives an example of a carry-ripple adder written in \revlang. Na\"{\i}vely compiling this implementation would result in a new bit being allocated for every carry bit, as the assignment on line $8$ is irreversible (note that $\texttt{carry_ex 1 1 0} = \texttt{1} = \texttt{carry_ex 1 1 1}$, hence the value of \texttt{c} can not be uniquely computed given \texttt{a}, \texttt{b} and the output). \rever reduces this space usage by automatically cleaning the old carry bit, allowing it to be reused.

\end{example}

\subsubsection{Semantics}

We designed the semantics of \revlang with two goals in mind:
\begin{enumerate}
\item keep the semantics as close to the original implementation as possible, and
\item simplify the task of formal verification.
\end{enumerate} 
The result is a somewhat non-standard semantics that is nonetheless intuitive for the programmer. Moreover, the particular semantics naturally enforces a style of programming that results in efficient circuits and allows common design patterns to be optimized.

The big-step semantics of \revlang is presented in Figure~\ref{fig:semantics} as a relation ${\implies} \subseteq \text{Config}\times\text{Config}$ on configuration-pairs -- pairs of terms and Boolean-valued stores. A key feature of our semantics is that Boolean, or bit values, are always allocated on the store. Specifically, Boolean constants and expressions are modelled by allocating a new location on the store to hold its value -- as a result all Boolean values, including constants, are mutable.

The allocation of Boolean values on the store serves two main purposes: to give the programmer fine-grain control over how many bits are allocated, and to provide a simple and efficient model of \emph{registers} -- i.e. arrays of bits. Specifically, registers are modelled as static length lists of bits. This allows the programmer to perform array-like operations such as bit modifications ($\rassign{\rindex{t_1}{i}}{t_2}$) as well as list-like operations such as slicing ($\rslice{t}{i}{j}$) and concatenation ($\rappend{t_1\!}{t_2}$) without copying out entire registers. We found that these were the most common access patterns for arrays of bits in low-level bitwise code (e.g. arithmetic and cryptographic implementations).

The semantics of $\oplus$ (Boolean XOR) and $\land$ (Boolean AND) are also notable in that they first reduce both arguments to locations, \emph{then} retrieve their value. This results in statements whose value may not be immediately apparent -- e.g., $\rxor{x}{(\rseq{\rassign{x}{y}}{y})}$, which under these semantics will always evaluate to \false. The benefit of this definition is that it allows the compiler to perform important optimizations without a significant burden on the programmer.

\begin{figure}[t]
\begin{minipage}{0.20\textwidth}
\begin{flalign*}
	{\bf Store} \;\;&\heap :{\N}\pmap \B & \\
	{\bf Config} \;\;&c ::= \config{t}{\heap} & \\
\end{flalign*}
\end{minipage}
\begin{minipage}{0.75\textwidth}

\centerline{
\inference[\sc{[let]}]
{
	\config{t_1}{\heap} => \config{v_1}{\heap'} \qquad
	\config{t_2\subst{x}{v_1}}{\heap'} => \config{v_2}{\heap''}
}{
	\config{\rlet{x}{t_1}{t_2}}{\heap} => \config{v_2}{\heap''}
}}
\end{minipage}
\vspace{-10pt}

\begin{minipage}{0.25\textwidth}
\centerline{
\inference[\sc {[refl]}]{}
{
	\config{v}{\heap} => \config{v}{\heap}
}}
\end{minipage}
\begin{minipage}{0.75\textwidth}
\centerline{
\inference[\sc{[bexp]}]
{
	\config{t_1}{\heap} => \config{l_1}{\heap'} \quad
	\config{t_2}{\heap'} => \config{l_2}{\heap''} \quad
	l_3\notin\dom(\heap'')
}{
	\config{t_1\star t_2}{\heap} => \config{l_3}{\heap''\subst{l_3}{\heap''(l_1)\star\heap''(l_2)}}
}}
\end{minipage}
\vspace{-10pt}

\begin{minipage}{0.45\textwidth}
\centerline{
\inference[\sc{[bool]}]
{
	b\in\B \quad l\notin\dom(\heap)
}{
	\config{b}{\heap} => \config{l}{\heap''\subst{l}{b}}
}}
\end{minipage}
\begin{minipage}{0.49\textwidth}
\vspace{1.5em}
\centerline{
\inference[\sc{[app]}]
{
	\config{t_1}{\heap} => \config{\rfun{x}{t_1'}}{\heap'}
	\config{t_2}{\heap'} => \config{v_2}{\heap''} \\
	\config{t_1'\subst{x}{v_2}}{\heap''} => \config{v}{\heap'''}
}{
	\config{\rapply{t_1}{t_2}}{\heap} => \config{v}{\heap'''}
}}
\end{minipage}
\vspace{-10pt}

\begin{minipage}{0.49\textwidth}
\centerline{
\inference[\sc{[seq]}]
{
	\config{t_1}{\heap} => \config{\runit}{\heap'} \\
	\config{t_2}{\heap'} => \config{v}{\heap''}
}{
	\config{\rseq{t_1}{t_2}}{\heap} => \config{v}{\heap''}
}}
\end{minipage}
\begin{minipage}{0.42\textwidth}
\vspace{1.5em}
\centerline{
\inference[\sc{[assn]}]
{
	\config{t_1}{\heap} => \config{l_1}{\heap'} \\
	\config{t_2}{\heap'} => \config{l_2}{\heap''}
}{
	\config{\rassign{t_1}{t_2}}{\heap} => \config{\runit}{\heap''\subst{l_1}{\heap''(l_2)}}
}}
\end{minipage}

\vspace{5pt}
\centerline{
\inference[\sc{[append]}]
{
	\config{t_1}{\heap} => \config{\rregister{l_1}{l_m}}{\heap'} \quad
	\config{t_2}{\heap'} => \config{\rregister{l_{m+1}}{l_n}}{\heap''}
}{
	\config{\rappend{t_1}{t_2}}{\heap} => \config{\rregister{l_1}{l_n}}{\heap''}
}}
\vspace{0.8em}

\begin{minipage}{0.58\textwidth}
\centerline{
\inference[\sc{[index]}]
{
	\config{t}{\heap} => \config{\rregister{l_1}{l_n}}{\heap'} \qquad
	1\leq i \leq n \qquad
}{
	\config{\rindex{t}{i}}{\heap} => \config{l_{i}}{\heap'}
}}
\vspace{5pt}
\centerline{
\inference[\sc{[slice]}]
{
	\config{t}{\heap} => \config{\rregister{l_1}{l_n}}{\heap'} \quad
	1\leq i\leq j\leq n
}{
	\config{\rslice{t}{i}{j}}{\heap} => \config{\rregister{l_i}{l_j}}{\heap'}
}}
\vspace{25pt}
\centerline{
\inference[\sc{[rotate]}]
{
	\config{t}{\heap} => \config{\rregister{l_1}{l_n}}{\heap'} \quad
	1< i < n
}{
	\config{\rrotate{t}{i}}{\heap} => \config{\rregister{l_{i}}{l_{i-1}}}{\heap'}
}}
\end{minipage}
\begin{minipage}{0.30\textwidth}
\vspace{-15pt}
\centerline{
\inference[\sc{[reg]}]
{
	\config{t_1}{\heap} => \config{l_1}{\heap_1} \\
	\config{t_2}{\heap} => \config{l_2}{\heap_2} \\
	\vdots \\
	\config{t_n}{\heap} => \config{l_n}{\heap_n}
}{
	\config{\rregister{t_1}{t_n}}{\heap} => \config{\rregister{l_1}{l_n}}{\heap_n}
}}
\end{minipage}

\vspace{5pt}
\centerline{
\inference[\sc{[clean]}]
{
	\config{t}{\heap} => \config{l}{\heap'} \quad
	\heap'(l) = \rfalse
}{
	\config{\rclean{t}}{\heap} => \config{\runit}{\restr{\heap'}{\dom(\heap')\setminus\{l\}}}
}
\quad
\inference[\sc{[assert]}]
{
	\config{t}{\heap} => \config{l}{\heap'} \quad
	\heap'(l) = \rtrue
}{
	\config{\rassert{t}}{\heap} => \config{\runit}{\heap'}
}}

\caption{Operational semantics of \revlang.}
\label{fig:semantics}
\end{figure}

\subsection{Boolean expressions}
Our compiler uses XOR-AND Boolean expressions -- single output classical circuits over XOR and AND gates -- as an intermediate language. Compilation from Boolean expressions into reversible circuits forms the main ``code generation'' step of our compiler. 

A Boolean expression is defined as an expression over Boolean constants, variable indices, and logical $\oplus$ and $\land$ operators. Explicitly, we define
$${\bf BExp} \;\;B ::= \false \| \true \| i\in\N \|  B_1 \oplus B_2 \| B_1 \land B_2.$$ Note that we use the symbols $\false, \true, \oplus$ and $\land$ interchangeably with their interpretation in $\B$. We use $\vars{B}$ to refer to the set of free variables in $B$. 

We interpret a Boolean expression as a function from (total) Boolean-valued states to Booleans. In particular, we define $\Stat = \N\rightarrow\B$ and denote the semantics of a Boolean expression by $\sem{B}:\Stat\rightarrow\B$. The formal definition of $\sem{B}$ is obvious so we omit it.

\subsection{Target architecture}
\rever compiles to \emph{combinational, reversible circuits} over NOT, controlled-NOT and Toffoli gates. By combinational circuits we mean a sequence of logic gates applied to bits with no external means of control or memory -- effectively pure logical functions. We chose this model as it is suitable for implementing classical functions and oracles within quantum computations \cite{NC:2000}.

Formally, we define
$${\bf Circ} \; C ::= \nil \| \notgate{i} \| \cnot{i}{j} \| \toffoli{i}{j}{k} \| C_1::C_2,$$ i.e., $\bf Circ$ is the free monoid over NOT, CNOT, and Toffoli gates with unit $\nil$ and the append operator $::$. All but the last bit in each gate is called a \emph{control}, whereas the final bit is denoted as the \emph{target} and is the only bit \emph{modified} or changed by the gate. We use $\uses{C}$, $\mods{C}$ and $\controls{C}$ to denote the set of bit indices that are used in, modified by, or used as a control in the circuit $C$, respectively. A circuit is \emph{well-formed} if no gate contains more than one reference to a bit -- i.e., the bits used in each controlled-NOT or Toffoli gate are distinct.

Similar to Boolean expressions, a circuit is interpreted as a function from states (maps from indices to Boolean values) to states, given by applying each gate which updates the previous state in order. The formal definition of the semantics of a reversible circuit $C$, given by $\sem{C}:\Stat\rightarrow\Stat$, is straightforward:
\begin{align*}
	\sem{\notgate{i}}\st &= s\subst{i}{\neg \st(i)} \\
	\sem{\cnot{i}{j}}\st &= s\subst{j}{\st(i)\oplus \st(j)} \\
	\sem{\toffoli{i}{j}{k}}\st &= s\subst{k}{(\st(i)\land \st(j))\oplus \st(k)} \\
	\sem{-}\st = \st \qquad& \sem{C_1::C_2}\st = (\sem{C_2}\circ \sem{C_1})\st
\end{align*}
We use $s\subst{x}{y}$ to denote the function that maps $x$ to $y$, and all other inputs $z$ to $s(z)$; by an abuse of notation we use $\subst{x}{y}$ to denote other substitutions as well.

\section{Compilation}\label{sec:compilation}

In this section we discuss the implementation of \rever. The compiler consists of around 4000 lines of code in a common subset of \fstar and \fsharp, with a front-end to evaluate and translate \fsharp quotations into \revlang expressions.

\subsection{Boolean expression compilation}

The core of \rever's code generation is a compiler from Boolean expressions into reversible circuits. We use a modification of the method employed in \revs.

As a Boolean expression is already in the form of an irreversible classical circuit, the main job of the compiler is to allocate ancillas to store sub-expressions whenever necessary. \rever does this by maintaining a (mutable) heap of ancillas $\ah\in\Anc$ called an \emph{ancilla heap}, which keeps track of the currently available (zero-valued) ancillary bits. Cleaned ancillas (ancillas returned to the zero state) may be pushed back onto the heap, and allocations return previously used ancillas if any are available, hence not using any extra space.

The function \textsc{compile-BExp}, shown in pseudo-code below, takes a Boolean expression $B$ and a target bit $i$ and then generates a reversible circuit computing $i \oplus B$. Note that ancillas are only allocated to store sub-expressions of $\land$ expressions, since $i\oplus (B_1\oplus B_2) = (i\oplus B_1)\oplus B_2$ and so we compile $i\oplus (B_1 \oplus B_2)$ by first computing $i'=i\oplus B_1$, followed by $i'\oplus B_2$.

\vspace{5pt}
{
\footnotesize
\begin{algorithmic}
\Function{compile-BExp}{$B$, $i$, $\ah$}
	\If{$B=0$} $-$
	\ElsIf{$B=1$} $\notgate{i}$
	\ElsIf{$B=j$} $\cnot{j}{i}$
	\ElsIf{$B= B_1 \oplus B_2$} \textsc{compile-BExp}($B_1$, $i$, $\ah$)::\textsc{compile-BExp}($B_2$, $i$, $\ah$)
	\Else \;// $B= B_1\land B_2$
		\State $a_1\gets$ pop-min($\ah$); $C\gets$ \textsc{compile-BExp}($B_1$, $a_1$, $\ah$);
		\State $a_2\gets$ pop-min($\ah$); $C'\gets$ \textsc{compile-BExp}($B_2$, $a_2$, $\ah$);
		\State $C::C'::\toffoli{a_1}{a_2}{i}$
	\EndIf
\EndFunction
\end{algorithmic}
}

\subsubsection{Cleanup}

The definition of \bcomp{} above leaves many garbage bits that take up space and need to be cleaned before they can be re-used. To reclaim those bits, we clean temporary expressions after every call to \bcomp{}.

To facilitate the cleanup -- or \emph{uncomputing} -- of a circuit, we define the \emph{restricted inverse} $\uncompute{C}{A}$ of $C$ with respect to a set of bits $A\subset\N$ by reversing the gates of $C$, and removing any gates with a target in $A$. For instance: \vspace{-0.5em}
$$\uncompute{\cnot{i}{j}}{A} = \begin{cases} \;\;- & \text{if } j\in A \\ \cnot{i}{j} & \text{otherwise} \end{cases}\vspace{-0.5em}$$
The other cases are defined similarly. Note that since uncompute produces a subsequence of the original circuit $C$, no ancillary bits are used.

The restricted inverse allows the temporary values of a reversible computation to be uncomputed without affecting any of the target bits. In particular, if $C=$ \bcomp{}($B, i$), then the circuit $C::\uncompute{C}{\{i\}}$ maps a state $s$ to $s\subst{i}{\sem{B}s\xor s(i)}$, allowing any newly allocated ancillas to be pushed back onto the heap. Intuitively, since no bits contained in the set $A$ are modified, the restricted inverse preserves their values; that the restricted inverse uncomputes the values of the remaining bits is less obvious, but it can be observed that if the computation doesn't \emph{depend} on the value of a bit in $A$, the computation will be inverted. We formalize and prove this statement in Section~\ref{sec:verification}.

\subsection{\revs compilation}

In studying the \revs compiler, we observed that most of what the compiler was doing was evaluating the non-Boolean parts of the program -- effectively bookkeeping for registers -- only generating circuits for a small kernel of cases. As a result, transformations to different Boolean representations (e.g., circuits, dependence graphs \cite{Parent:2015}) and the interpreter itself reused significant portions of this bookkeeping code. To make use of this redundancy to simplify both writing and verifying the compiler, we designed \rever as a \emph{partial evaluator} parameterized by an abstract machine for evaluating Boolean expressions. As a side effect, we arrive at a unique model for circuit compilation similar to staged computation (see, e.g., \cite{Jones:1993}).

\rever works by evaluating the program with an abstract machine providing mechanisms for initializing and assigning locations on the store to Boolean expressions. We call an instantiation of this abstract machine an \emph{interpretation} $\I$, which consists of a domain $D$ equipped with two operators:
\begin{align*}
	&\assign: D\times \N\times{\bf BExp}\rightarrow D \\
	&\eval: D\times \N\times\Stat\pmap \B.
\end{align*}

We typically denote an element of an interpretation domain $D$ by $\heap$. A sequence of assignments in an interpretation builds a Boolean computation or circuit within a specific model (i.e., classical, reversible, different gate sets) which may be simulated on an initial state with the \eval{} function -- effectively an operational semantics of the model. Practically speaking, an element of $D$ abstracts the store in Figure~\ref{fig:semantics} and allows delayed computation or additional processing of the Boolean expression stored in a cell, which may be mapped into reversible circuits immediately or after the entire program has been evaluated. We give some examples of interpretations below.

\begin{example}
The standard interpretation $\I_{standard}$ has domain $\Heap=\N\pmap\B$, together with the operations
\begin{align*}
	&\assign_{standard}(\heap, l, B)=\heap\subst{l}{\sem{B}\heap} \\
	&\eval_{standard}(\heap, l, \st) = \heap(l).
\end{align*}
Partial evaluation over the standard interpretation coincides exactly with the operational semantics of \revlang.
\end{example}

\begin{example}
The \emph{reversible circuit} interpretation $\I_{circuit}$ has domain $D_{circuit}=(\N\pmap\N)\times\Circ\times\Anc$. In particular, given $(\rho, C, \ah)\in D_{circuit}$, $\rho$ maps heap locations to bits in $C$, and $\ah$ is an ancilla heap. Assignment and evaluation are further defined as follows:
\begin{align*}
	&\assign_{circuit}((\rho, C, \ah), l, B)=(\rho\subst{l}{i}, C::C', \ah) \\
		&\qquad\text{where }i=\text{ pop-min}(\ah), \\ 
		&\qquad\hspace{2pt}(C', \ah')=\textsc{ compile-BExp}\left(B\subst{l'\in \vars{B}}{\rho(l')}, i, \ah\right) \\
	&\eval_{circuit}((\rho, C, \ah), l, \st)) = \left(\sem{C}\st\right)(\rho(l))
\end{align*}
Interpreting a program with $\I_{circuit}$ builds a reversible circuit executing the program, together with a mapping from heap locations to bits. Since the circuit is required to be reversible, when a location is overwritten, a new ancilla $i$ is allocated and the expression $B\oplus i$ is compiled into a circuit. Evaluation amounts to running the circuit on an initial state, then retrieving the value at the bit associated with a heap location. 
\end{example}

Given an interpretation $\I$ with domain $D$, we define the set of $\I$-configurati\-ons as $\mathbf{Config}_\I=\mathbf{Term}\times D$ -- that is, $\I$-configurations are pairs of programs and elements of $D$ which function as an abstraction of the heap. The relation $${\implies_\I}\subseteq\mathbf{Config}_\I\times\mathbf{Config}_\I$$ gives the operational semantics of \revs over the interpretation $\I$. We do not give a formal definition of $\implies_\I$, as it can be obtained trivially from the definition of $\implies$ (Figure~\ref{fig:semantics}) by replacing all heap updates with $\assign{}$ and taking $\dom(\heap)$ to mean the set of locations on which $\eval{}$ is defined. To compile a program term $t$, \rever evaluates $t$ over a particular interpretation $\I$ (for instance, the reversible circuit interpretation) and an initial heap $\sigma\in D$ according to the semantic relation $\implies_\I$. In this way, evaluating a program and compiling a program to a circuit look almost identical. This greatly simplifies the problem of verification (see Section~\ref{sec:verification}).

\rever currently supports three modes of compilation, defined by giving interpretations: a default mode, an eagerly cleaned mode, and a ``crush'' mode. The default mode evaluates the program using the circuit interpretation, and simply returns the circuit and output bit(s), while the eager cleanup mode operates analogously, using instead the garbage-collected interpretation defined below in Section~\ref{sec:gc}. The crush mode interprets a program as a list of Boolean expressions over free variables, which while unscalable allows highly optimized versions of small circuits to be compiled, a common practice in circuit synthesis. We omit the details of the Boolean expression interpretation. 

\subsubsection{Function compilation}

While the definition of \rever as a partial evaluator streamlines both development and verification, there is an inherent disconnect between the treatment of a (top-level) function expression by the interpreter and by the compiler, in that we want the compiler to evaluate the function body. Instead of defining a two-stage semantics for \revs we took the approach of applying a program transformation, whereby the function being compiled is evaluated on special heap locations representing the parameters. This creates a further problem in that the compiler needs to first determine the size of each parameter; to solve this problem, \rever performs a static analysis we call \emph{parameter interference}. We omit the details of this analysis due to space constraints and instead point the interested reader to an extended version of this paper \cite{Amy:2016}.

\subsection{Eager cleanup}\label{sec:gc}

It was previously noted that the circuit interpretation allocates a new ancilla on every assignment to a location, due to the requirement of reversibility. Apart from \rever's additional optimization passes, this is effectively the Bennett method, and hence uses a large amount of extra space. One way to keep the space usage from continually expanding as assignments are made is to clean the old bit as soon as possible and then reuse it, rather than wait until the end of the computation. Here we develop an interpretation that performs this automatic, eager cleanup by augmenting the circuit interpretation with a \emph{cleanup expression} for each bit. Our method is based on the eager cleanup of \cite{Parent:2015}, and was intended as a more easily verifiable alternative to mutable dependency diagrams.

The \emph{eager cleanup} interpretation $\I_{GC}$ has domain $$D=(\N\pmap\N)\times\Circ\times\Anc\times(\N\pmap\BExp),$$ where given $(\rho, C, \ah, \cl)\in D$, $\rho$, $C$ and $\ah$ are as in the circuit interpretation. The partial function $\cl$ maps individual bits to a Boolean expression over the bits of $C$ which can be used to return the bit to its initial state, called the cleanup expression. Specifically, we have the following property: 
$$\forall i\in\cod(\rho), s'(i) \oplus \sem{\cl(i)}s' = s(i) \qquad\text{where $\st'=\sem{C}\st$.}$$ Intuitively, any bit $i$ can then be cleaned by simply computing $i\mapsto i\oplus \cl(i)$, which in turn can be done by calling \textsc{compile-BExp}($\cl(i)$, $i$).

Two problems remain, however. In general it may be the case that a bit \emph{can not} be cleaned without affecting the value of other bits, as it might result in a loss of information -- in the context of cleanup expressions, this occurs exactly when a bit's cleanup expression contains an irreducible self-reference. In particular, if $i\in\vars{B}$, then \textsc{compile-BExp}($B$, $i$) does not compile a circuit computing $i\oplus B$ and hence won't clean the target bit correctly. In the case when a garbage bit contains a self-reference in its cleanup expression that can not be eliminated by Boolean simplification, \rever simply ignores the bit and performs a final round of cleanup at the end.

The second problem arises when a bit's cleanup expression references another bit that has itself since been cleaned or otherwise modified. In this case, the modification of the latter bit has invalidated the correctness property for the former bit. To ensure that the above relation always holds, whenever a bit is modified -- corresponding to an XOR of the bit, $i$, with a Boolean expression $B$ -- all instances of bit $i$ in every cleanup expression is replaced with $i\oplus B$. Specifically we observe that, if $s'(i) = s(i)\oplus \sem{B}s$, then $$s'(i)\oplus \sem{B}s = s(i)\oplus \sem{B}s \oplus\sem{B}s = s(i).$$

The function $\textsc{clean}$, defined below, performs the cleanup of a bit $i$ if possible, and validates all cleanup expressions in a given element of $D$:

{
\footnotesize
\begin{algorithmic}
\Function{clean}{$(\rho, C, \ah, \cl)$, $i$}
	\If{$i\in\vars{\cl(i)}$} \Return $(\rho, C, \ah, \cl)$
	\Else
		\State $C'\gets\text{\textsc{compile-BExp}}(\cl(i), i, \ah)$
		\State \textbf{if} $i$ is an ancilla \textbf{then} insert($i, \ah$)
		\State $\cl'\gets\cl\subst{i'\in\dom(\cl)}{\cl(i')\subst{i}{i\oplus\cl(i)}}$
		\State \Return $(\rho, C::C', \ah, \cl')$
	\EndIf
\EndFunction
\end{algorithmic}
}

Assignment and evaluation are defined in the eager cleanup interpretation as follows. Both are effectively the same as in the circuit interpretation, except the assignment operator calls \textsc{clean} on the previous bit mapped to $l$.
\begin{align*}
	&\assign_{GC}((\rho, C, \ah, \cl), l, B)=\textsc{clean}((\rho\subst{l}{i},C::C', \ah, \cl\subst{i}{B'}), i) \\
		&\quad\text{where $i=$ pop-min($\ah$),} \\ 
		&\qquad\quad\text{$B'=B\subst{l'\in \vars{B}}{\rho(l')}$} \\
		&\qquad\quad\text{$C'=$\textsc{ compile-BExp}($B', i, \ah)$} \\
	&\eval_{GC}((\rho, C, \ah, \cl), l, \st)) = \left(\sem{C}\st\right)(\rho(l))
\end{align*}

The eager cleanup interpretation coincides with a reversible analogue of \emph{garbage collection} for a very specific case when the number of references to a heap location (or in our case, a bit) is trivially zero. In fact, the \textsc{clean} function can be used to eagerly clean bits that have no reference in other contexts. We intend to expand \rever to include a generic garbage collector that uses cleanup expressions to more aggressively reclaim space -- for instance, when a bit's unique pre-image on the heap leaves the current scope.

\subsection{Optimizations}

During the course of compilation it is frequently the case that more ancillas are allocated than are actually needed, due to the program structure. For instance, when compiling the expression $i\gets B$, if $B$ can be factored as $i\oplus B'$ the assignment may be performed reversibly rather than  allocating a new bit to store the value of $B$. Likewise if $i$ is provably in the $0$ or $1$ state, the assignment may be performed reversibly without allocating a new bit. Our implementation identifies some of these common patterns, as well as general Boolean expression simplifications, to further minimize the space usage of compile circuits. All such optimizations in \rever have been formally verified.

\section{Verification}\label{sec:verification}

In this section we describe the formal verification of \rever and give the major theorems proven. All theorems given in this section have been formally specified and proven using the \fstar compiler \cite{Swamy:2015}. We first give theorems about our Boolean expression compiler, then use these to prove properties about whole program compilation. The total verification of the \rever core's approximately 2000 lines of code comprises around 2200 lines of \fstar code, and took just over 1 person-month. We feel that this relatively low-cost verification is a testament to the increasing ease with which formal verification can be carried out using modern proof assistants. Additionally, the verification relies on only 11 unproven axioms regarding simple properties of lookup tables and sets, such as the fact that a successful lookup is in the codomain of a lookup table.

Rather than give \fstar specifications, we translate our proofs to mathematical language as we believe this is more enlightening. The full source code of \rever including proofs can be obtained at \href{https://github.com/msr-quarc/ReVerC}{https://github.com/msr-quarc/ReVerC}.

\subsection{Boolean expression compilation}

\subsubsection{Correctness}

Below is our main theorem establishing the correctness of the function $\bcomp$ with respect to the semantics of reversible circuits and Boolean expressions. It states that if the variables of $B$, the bits on the ancilla heap and the target are non-overlapping, and if the ancilla bits are $0$-valued, then the circuit computes the expression $i\oplus B$.

\begin{theorem}\label{thm:bexpcorrect}
Let $B$ be a Boolean expression, $\ah$ be an ancilla heap, $i\in\N$, $C\in\Circ$ and $\st$ be a map from bits to Boolean values. Suppose $\vars{B}$, $\ah$ and $\{i\}$ are all disjoint and $\st(j)=0$ for all $j\in\ah$. Then $$\left(\sem{\bcomp(B, i, \ah)}\st\right)(i)=\st(i)\oplus\sem{B}\st.$$
\end{theorem}

\subsubsection{Cleanup}

As remarked earlier, a crucial part of reversible computing is cleaning ancillas both to reduce space usage, and in quantum computing to prevent entangled qubits from influencing the computation. Moreover, the correctness of our cleanup is actually necessary to prove correctness of the compiler, as the compiler re-uses cleaned ancillas on the heap, potentially interfering with the precondition of Theorem~\ref{thm:bexpcorrect}. We use the following lemma to establish the correctness of our cleanup method, stating that the uncompute transformation reverses all changes on bits not in the target set under the condition that no bits in the target set are used as controls.

\begin{lemma}\label{thm:cleanup}
Let $C$ be a well-formed reversible circuit and $A\subset\N$ be some set of bits. If $A\cap \controls{C}=\emptyset$ then for all states $\st, \st'=\sem{C::\uncompute{C}{A}}\st$ and any $i\notin A$, $$\st(i) = \st'(i)$$
\end{lemma}
Lemma~\ref{thm:cleanup} largely relies on the following important lemma stating in effect that the action of a circuit is determined by the values of the bits used as controls:
\begin{lemma}\label{lem:stateswap}
Let $A\subset\N$ and $\st,\st'$ be states such that for all $i\in A$, $\st(i)=\st'(i)$. If $C$ is a reversible circuit where $\controls{C}\subseteq A$, then $$(\sem{C}\st)(i)=(\sem{C}\st')(i)$$ for all $i\in A$.
\end{lemma}

Lemma~\ref{thm:cleanup}, together with the fact that $\bcomp$ produces a well-formed circuit under disjointness constraints, gives us our cleanup theorem below that Boolean expression compilation with cleanup correctly reverses the changes to every bit except the target.

\begin{theorem}
Let $B$ be a Boolean expression, $\ah$ be an ancilla heap and $i\in\N$ such that $\vars{B}$, $\ah$ and $\{i\}$ are all disjoint. Suppose $\bcomp(B, i, \ah)=C$. Then for all $j\neq i$ and states $\st$ we have $$\left(\sem{C\circ\uncompute{C}{\{i\}}}\st\right)(j)=\st(j).$$
\end{theorem}

\subsection{\revlang compilation}

It was noted in Section~\ref{sec:compilation} that the design of \rever as a partial evaluator simplifies proving correctness. We expand on that point now, and in particular show that if a relation between elements of two interpretations is preserved by assignment, then the evaluator also preserves the relation. We state this formally in the theorem below.

\begin{theorem}\label{thm:sim}
Let $\I_1, \I_2$ be interpretations and suppose whenever $(\heap_1, \heap_2)\in R$ for some relation $R\subseteq \I_1\times\I_2$, $$(\assign_{1}(\heap_1, l, B), \assign_{2}(\heap_2, l, B))\in R$$ for any $l, B$. Then for any term $t$, if $\config{t}{\heap_1}\implies_{\I_1}\config{v_1}{\heap_1'}$ and $\config{t}{\heap_2}\implies_{\I_2}\config{v_2}{\heap_2'}$, then $v_1=v_2$ and $(\heap_1',\heap_2')\in R$.
\end{theorem}

Theorem~\ref{thm:sim} lifts properties about interpretations to properties of evaluation over those abstract machines -- in particular, we only need to establish that \emph{assignment} is correct for an interpretation to establish correctness of the corresponding evaluator/compiler. In practice we found this significantly reduces boilerplate proof code that is otherwise currently necessary in \fstar due to a lack of automated induction.

Given two interpretations $\I, \I'$, we say elements $\heap$ and $\heap'$ of $\I$ and $\I'$ are \emph{observationally equivalent} with respect to a supplied set of initial values $\st\in\Stat$ if for all $i\in\N$, $\eval_{\I}(\heap, i, \st) = \eval_{\I'}(\heap', i, \st)$. We say $\heap\sim_\st\heap'$ if $\heap$ and $\heap'$ are observationally equivalent with respect to $\st$. As observational equivalence of two domain elements $\heap, \heap'$ implies that any location in scope has the same valuation in either interpretation, it suffices to show that any compiled circuit is observationally equivalent to the standard interpretation. The following lemmas are used along with Theorem~\ref{thm:sim} to establish this fact for the default and eager-cleanup interpretations -- a similar lemma is proven in the implementation of \rever for the crush mode.

\begin{lemma}\label{thm:assign}
Let $\heap, \heap'$ be elements of $\I_{standard}$ and $\I_{circuit}$, respectively. For all $l\in\N, B\in\BExp, \st\in\Stat$, if $\heap\sim_\st\heap'$ and $\st(i)=0$ whenever $i\in\ah$, then $$\assign_{standard}(\heap, l, B)\sim_\st\assign_{circuit}(\heap', l, B).$$ Moreover, the ancilla heap remains $0$-filled.
\end{lemma}

We say that $(\rho, C, \ah)\in D_{circuit}$ is \emph{valid} with respect to $s\in\Stat$ if and only if $s(i) = 0$ for all $i\in\ah$. For elements of $D_{GC}$ the validity conditions are more involved, so we introduce a relation, $\mathcal{V}\subseteq D_{GC}\times \Stat$, defining the set of valid domain elements:\vspace{-5pt}
\begin{align*}
	((\rho, C, \ah, \cl), \st)\in\mathcal{V} \iff \;
		\forall i\in\ah, s(i) = 0 &\land \forall l, l'\in\dom(\rho), \rho(l) \neq \rho(l') \\
		&\land\forall i\in\cod(\rho), \sem{i\oplus \cl(i)}(\sem{C}s) = s(i)\vspace{-5pt}
\end{align*}
Informally, $\mathcal{V}$ specifies that all bits on the heap have initial value $0$, that $\rho$ is a one-to-one mapping, and that for every active bit $i$, XORing $i$ with $\cl(i)$ returns the initial value of $i$ -- that is, $i\oplus \cl(i)$ \emph{cleans} $i$.

\begin{lemma}\label{thm:assignGC}
Let $\heap, \heap'$ be elements of $\I_{standard}$ and $\I_{GC}$, respectively. For all $l\in\N, B\in\BExp, \st\in\Stat$, if $\heap\sim_\st\heap'$ and $(\heap', \st)\in\mathcal{V}$, then $$\assign_{standard}(\heap, l, B)\sim_\st\assign_{GC}(\heap', l, B).$$ Moreover, $(\assign_{GC}(\heap', l, B), \st)\in\mathcal{V}$.
\end{lemma}

By setting the relation $R_{GC}$ as \vspace{-5pt} $$(\sigma_1, \sigma_2)\in R_{GC} \iff \sigma_2\in \mathcal{V} \land \sigma_1 \sim_{s_0} \sigma_2\vspace{-5pt}$$ for $\sigma_1\in D_{standard}$, by Theorem~\ref{thm:sim} and Lemma~\ref{thm:assignGC} it follows that partial evaluation/compilation preserves observational equivalence between $\I_{standard}$ and $\I_{GC}$. A similar result follows for $\I_{circuit}$.

To formally prove correctness of the compiler we need initial values in each interpretation (and an initial state) which are observationally equivalent. We don't describe the initial values here as they are dependent on the program transformation applied to expand top-level functions.
\vspace{-8pt}

\section{Experiments}\label{sec:experiments}

\vspace{-5pt}

We ran experiments to compare the bit, gate and Toffoli counts of circuits compiled by \rever to the original \revs compiler. The number of Toffoli gates in particular is distinguished as such gates are generally much more costly than NOT and controlled-NOT gates -- at least $7$ times as typical implementations use $7$ CNOT gates \emph{per Toffoli} \cite{NC:2000}, or up to hundreds of times in most fault-tolerant architectures \cite{Amy:2013}.
We compiled circuits for various arithmetic and cryptographic functions written in \revlang using both compilers and reported the results in Table~\ref{tab:def}. Experiments were run in Linux using 8\,GB of RAM.

\begin{table}
\scriptsize
\centering
\begin{tabular}{ l  r  r  r  r  r  r  r r r  r r r}
	\toprule
	Benchmark & \multicolumn{3}{c}{\revs} & \multicolumn{3}{c}{\revs (eager)} & \multicolumn{3}{c}{\rever}  & \multicolumn{3}{c}{\rever (eager)} \\ \cmidrule(l{4pt}r{1pt}){2-4} \cmidrule(l{4pt}r{1pt}){5-7} \cmidrule(l{4pt}r{1pt}){8-10} \cmidrule(l{4pt}r{1pt}){11-13}
	 & bits & gates & \hspace{-1pt}Toffolis & bits & gates & \hspace{-1pt}Toffolis & bits & gates & \hspace{-1pt}Toffolis & bits & gates & \hspace{-1pt}Toffolis \\ \midrule
	carryRippleAdd 32 
		& 129 & 281 & {\bf 62} 
		& 129 & 467 & 124 
		& 128 & 281 & {\bf 62}
		& {\bf 113} & 361 & 90 \\
	carryRippleAdd 64 
		& 257 & 569 & {\bf 126} 
		& 257 & 947 & 252
		& 256 & 569 & {\bf 126}
		& {\bf 225} & 745 & 186 \\ 
	mult 32 
		& 128 & 6016 & 4032 
		& 128 & 6016 & 4032
		& 128 & 6016 & 4032
		& 128 & 6016 & 4032 \\ 
	mult 64 
		& 256 & 24320 & 16256
		& 256 & 24320 & 16256 
		& 256 & 24320 & 16256
		& 256 & 24320 & 16256 \\ 
	carryLookahead 32 
		& 160 & 345 & {\bf 103} 
		& {\bf 109} & 1036 & 344
		& 165 & 499 & 120
		& 146 & 576 & 146 \\
	carryLookahead 64 
		& 424 & 1026 & {\bf 307} 
		& {\bf 271} & 3274 & 1130
		& 432 & 1375 & 336
		& 376 & 1649 & 428 \\ 
	modAdd 32 
		& 65 & 188 & 62 
		& 65 & 188 & 62
		& 65 & 188 & 62
		& 65 & 188 & 62 \\
	modAdd 64 
		& 129 & 380 & 126 
		& 129 & 380 & 126 
		& 129 & 380 & 126
		& 129 & 380 & 126  \\ 
	cucarroAdder 32 
		& 65 & 98 & 32
		& 65 & 98 & 32 
		& 65 & 98 & 32
		& 65 & 98 & 32 \\
	cucarroAdder 64 
		& 129 & 194 & 64 
		& 129 & 194 & 64
		& 129 & 194 & 64
		& 129 & 194 & 64 \\
	ma4 
		& 17 & 24 & 8
		& 17 & 24 & 8
		& 17 & 24 & 8
		& 17 & 24 & 8 \\
	SHA-2 round 
		& 449 & 1796 & {\bf 594} 
		& {\bf 353} & 2276 & 754
		& 452 & 1796 & {\bf 594}
		& 449 & 1796 & {\bf 594} \\
	MD5 
		& 7841 & 81664 & {\bf 27520}
		& 7905 & 82624 & 27968
		& 4833 & 70912 & {\bf 27520}
		& {\bf 4769} & 70912 & {\bf 27520} \\ 
	\bottomrule
\end{tabular}
\caption{Bit and gate counts for both compilers in default and eager cleanup modes. In cases when not all results are the same, entries with the fewest bits used or Toffolis are bolded.}
\label{tab:def}
\end{table}

The results show that both compilers are more-of-less evenly matched in terms of bit counts across both modes, despite \rever being certifiably correct. \rever's eager cleanup mode never used more bits than the default mode, as expected, and in half of the benchmarks reduced the number of bits. Moreover, in the cases of the carryRippleAdder and MD5 benchmarks, \rever's eager cleanup mode produced circuits with significantly fewer bits than either of \revs' modes. On the other hand, \revs saw dramatic decreases in bit numbers for carryLookahead and SHA-2 with its eager cleanup mode compared to \rever.

While the results show there is clearly room for optimization of gate counts, they appear consistent with other verified compilers (e.g., \cite{Leroy:2006}) which take some performance hit when compared to unverified compilers. In particular, unverified compilers may use more aggressive optimizations due to the increased ease of implementation and the lack of a requirement to prove their correctness compared to certified compilers. In some cases, the optimizations are even known to not be correct in all possible cases, as in the case of fast arithmetic and some loop optimization passes in the GNU C Compiler \cite{GCC}.


\section{Conclusion}\label{sec:conclusion}

We have described our verified compiler for the \revs language, \rever. Our method of compilation differs from the original \revs compiler by using partial evaluation over an interpretation of the heap to compile programs, forgoing the need to re-implement and verify bookkeeping code for every internal translation. We described two interpretations implemented in \rever, the circuit interpretation and a garbage collected interpretation, the latter of which refines the former by applying eager cleanup.

While \rever is verified in the sense that compiled circuits produce the same result as the program interpreter, as with any verified compiler project this is not the end of certification. The implementation of the interpreter may have subtle bugs, which ideally would be verified against a more straightforward adaptation of the semantics using a relational definition. We intend to address these issues in the future, and to further improve upon \rever's space usage.

\bibliographystyle{splncs03}
\bibliography{main}

\begin{thebibliography}{10}
\providecommand{\url}[1]{\texttt{#1}}
\providecommand{\urlprefix}{URL }

\bibitem{GCC}
Using the {GNU} {C}ompiler {C}ollection. Free Software Foundation, Inc. (2016),
  \url{https://gcc.gnu.org/onlinedocs/gcc/}

\bibitem{Amy:2013}
Amy, M., Maslov, D., Mosca, M., Roetteler, M.: A meet-in-the-middle algorithm
  for fast synthesis of depth-optimal quantum circuits. IEEE Transactions on
  Computer Aided Design of Integrated Circuits and Systems  32(6),  818--830
  (2013)

\bibitem{Amy:2016}
Amy, M., Roetteler, M., Svore, K.M.: Verified compilation of space-efficient
  reversible circuits. arXiv e-prints  (2016),
  \url{https://arxiv.org/abs/1603.01635}

\bibitem{Zuliani16}
Anticoli, L., Piazza, C., Taglialegne, L., Zuliani, P.: Towards quantum
  programs verification: From {Q}uipper circuits to {QPMC}. In: Proceedings of
  the 8th international Conference on Reversible Computation (RC'16). pp.
  213--219 (2016)

\bibitem{Bennett:73}
Bennett, C.H.: Logical reversibility of computation. IBM Journal of Research
  and Development  17,  525--532 (1973)

\bibitem{Beringer:2014}
Beringer, L., Stewart, G., Dockins, R., Appel, A.: Verified compilation for
  shared-memory {C}. In: Programming Languages and Systems, vol. 8410, pp.
  107--127. Springer LNCS (2014)

\bibitem{Chlipala:2010}
Chlipala, A.: A verified compiler for an impure functional language. In:
  Proceedings of the 37th Annual ACM SIGPLAN-SIGACT Symposium on Principles of
  Programming Languages (POPL'10). pp. 93--106. ACM (2010)

\bibitem{Claessen:2001}
Claessen, K.: Embedded Languages for Describing and Verifying Hardware. {PhD}
  thesis, Chalmers University of Technology and G\"{o}teborg University (2001)

\bibitem{Fournet:2013}
Fournet, C., Swamy, N., Chen, J., Dagand, P.E., Strub, P.Y., Livshits, B.:
  Fully abstract compilation to javascript. In: Proceedings of the 40th Annual
  ACM SIGPLAN-SIGACT Symposium on Principles of Programming Languages
  (POPL'13). pp. 371--384. ACM (2013)

\bibitem{GLR+:2013b}
Green, A.S., {LeFanu Lumsdaine}, P., Ross, N.J., Selinger, P., Valiron, B.:
  {Quipper: a scalable quantum programming language}. In: Proceedings of the
  34th annual ACM SIGPLAN conference on Programming Language Design and
  Implementation (PLDI'13). ACM (2013)

\bibitem{Grover:96}
Grover, L.K.: A fast quantum mechanical algorithm for database search. In:
  Proceedings of the 28th Annual ACM Symposium on the Theory of Computing
  (STOC'96). pp. 212--219. ACM (1996)

\bibitem{Jones:1993}
Jones, N.D., Gomard, C.K., Sestoft, P.: Partial Evaluation and Automatic
  Program Generation. Prentice-Hall, Inc., Upper Saddle River, NJ, USA (1993)

\bibitem{Leroy:2006}
Leroy, X.: Formal certification of a compiler back-end or: Programming a
  compiler with a proof assistant. In: Proceedings of the 34th Annual ACM
  SIGPLAN-SIGACT Symposium on Principles of Programming Languages (POPL'06).
  pp. 42--54. ACM (2006)

\bibitem{LJ:2014}
Lin, C.C., Jha, N.K.: {RMDDS: Reed-Muller decision diagram synthesis of
  reversible logic circuits}. Journal on Emerging Technologies in Computing
  Systems  10(2), ~14 (2014)

\bibitem{Markov:2014}
Markov, I.L.: Limits on fundamental limits to computation. Nature  512,
  147--154 (2014)

\bibitem{MMD:2007}
Maslov, D., Miller, D.M., Dueck, G.W.: {Techniques for the synthesis of
  reversible Toffoli networks}. ACM Transactions on Design Automation of
  Electronic Systems  12(4), ~42 (2007)

\bibitem{MMD:2003}
Miller, D.M., Maslov, D., Dueck, G.W.: A transformation based algorithm for
  reversible logic synthesis. In: Proceedings of the 40th Annual Design
  Automation Conference (DAC'03). pp. 318--323 (2003)

\bibitem{Neis:2015}
Neis, G., Hur, C.K., Kaiser, J.O., McLaughlin, C., Dreyer, D., Vafeiadis, V.:
  Pilsner: A compositionally verified compiler for a higher-order imperative
  language. In: Proceedings of the 20th ACM SIGPLAN International Conference on
  Functional Programming (ICFP'15). pp. 166--178. ACM (2015)

\bibitem{NC:2000}
Nielsen, M.A., Chuang, I.L.: Quantum Computation and Quantum Information.
  Cambridge University Press, Cambridge, UK (2000)

\bibitem{Omer:2000}
{\"O}mer, B.: {Quantum programming in QCL}. Master's thesis, Technical
  University of Vienna (2000)

\bibitem{Parent:2015}
Parent, A., Roetteler, M., Svore, K.M.: Reversible circuit compilation with
  space constraints. arXiv e-prints  (2015),
  \url{https://arxiv.org/abs/1510.00377}

\bibitem{Perconti:2014}
Perconti, J., Ahmed, A.: Verifying an open compiler using multi-language
  semantics. In: ACM Transactions on Programming Languages and Systems, vol.
  8410, pp. 128--148. Springer LNCS (2014)

\bibitem{SM:2013}
Saeedi, M., Markov, I.L.: Synthesis and optimization of reversible circuits.
  ACM Computing Surveys  45(2), ~21 (2013)

\bibitem{Scherer:2015}
{Scherer}, A., {Valiron}, B., {Mau}, S.C., {Alexander}, S., {van den Berg}, E.,
  {Chapuran}, T.E.: {Resource analysis of the quantum linear system algorithm}.
  arXiv e-prints  (2015), \url{https://arxiv.org/abs/1505.06552}

\bibitem{SSP:2013}
Shafaei, A., Saeedi, M., Pedram, M.: Reversible logic synthesis of $k$-input,
  $m$-output lookup tables. In: Proceedings of the Conference on Design,
  Automation and Test in Europe (DATE'13). pp. 1235--1240 (2013)

\bibitem{Shor:97}
Shor, P.W.: Polynomial-time algorithms for prime factorization and discrete
  logarithms on a quantum computer. SIAM Journal on Computing  26(5),
  1484--1509 (1997)

\bibitem{Swamy:2015}
Swamy, N., Hri\c{t}cu, C., Keller, C., Rastogi, A., Delignat-Lavaud, A.,
  Forest, S., Bhargavan, K., Fournet, C., Strub, P.Y., Kohlweiss, M.,
  Zinzindohoue, J.K., Zanella-B{\'e}guelin, S.: Dependent types and
  multi-monadic effects in {F$^\star$}. In: Proceedings of the 43rd Annual ACM
  SIGPLAN-SIGACT Symposium on Principles of Programming Languages (POPL'16).
  pp. 256--270. ACM (2016)

\bibitem{Thomsen:2012}
Thomsen, M.K.: A functional language for describing reversible logic. In:
  Proceedings of the 2012 Forum on Specification and Design Languages (FDL'12).
  pp. 135--142. IEEE (2012)

\bibitem{WD:2010}
Wille, R., Drechsler, R.: Towards a Design Flow for Reversible Logic. Springer
  (2010)

\bibitem{Wille:2009}
Wille, R., Grosse, D., Miller, D., Drechsler, R.: Equivalence checking of
  reversible circuits. In: Proceedings of the 39th IEEE International Symposium
  on Multiple-Valued Logic (ISMVL'09). pp. 324--330 (2009)

\bibitem{Wille:2010}
Wille, R., Offermann, S., Drechsler, R.: Syrec: A programming language for
  synthesis of reversible circuits. In: Proceedings of the 2010 Forum on
  Specification and Design Languages (FDL'10). pp. 1--6 (2010)

\bibitem{Yamashita:2010}
Yamashita, S., Markov, I.: Fast equivalence-checking for quantum circuits. In:
  Proceedings of the 2010 IEEE/ACM Symposium on Nanoscale Architectures
  (NANOARCH'10). pp. 23--28 (2010)

\bibitem{yokoyama2007reversible}
Yokoyama, T., Gl{\"u}ck, R.: A reversible programming language and its
  invertible self-interpreter. In: Proceedings of the 2007 Symposium on Partial
  Evaluation and Semantics-Based Program Manipulation (PEPM'07). pp. 144--153.
  ACM (2007)

\end{thebibliography}

\end{document}